\documentclass[11 pt]{amsart}
\usepackage{amsmath, amsfonts}

\usepackage{defsnew}
\usepackage{macros_Chinmay}
\usepackage{notation_mw}
\usepackage{comment}
\usepackage{tikz}
\usetikzlibrary{positioning}
\usepackage{amsthm}
\usepackage{setspace}
\usepackage{cases}
\usepackage{hyperref}
\usepackage{defs_Routing}

\title{Inducing Social Optimality in Games via Adaptive Incentive Design}

\author[Maheshwari, Kulkarni, Wu, Sastry]{Chinmay Maheshwari\({}^{1}\), Kshitij Kulkarni\({}^{1}\), Manxi Wu\({}^{1,2}\), and S. Shankar Sastry\({}^{1}\)
}
\thanks{\(^{1}\) The authors are with the Department of Electrical Engineering and Computer Sciences, University of California Berkeley, USA. \href{mailto:chinmay\_maheshwari@berkeley.edu}{chinmay\_maheshwari@berkeley.edu}, \href{mailto:kshitijkulkarni@berkeley.edu}{kshitijkulkarni@berkeley.edu} \href{mailto:manxiwu@berkeley.edu}{manxiwu@berkeley.edu}, \href{mailto:sastry@eecs.berkeley.edu}{sastry@eecs.berkeley.edu}.}
\thanks{\({}^2\) The author is with Simons Institute for the Theory of Computing, Berkeley
}

\begin{document}

\maketitle

\begin{abstract}
How can a social planner adaptively incentivize selfish agents who are learning in a strategic environment to induce a socially optimal outcome in the long run? We propose a two-timescale learning dynamics to answer this question in both atomic and non-atomic games. In our learning dynamics, players adopt a class of learning rules to update their strategies at a faster timescale, while a social planner updates the incentive mechanism at a slower timescale. In particular, the update of the incentive mechanism is based on each player's externality, which is evaluated as the difference between the player's marginal cost and the society's marginal cost in each time step. We show that any fixed point of our learning dynamics corresponds to the optimal incentive mechanism such that the corresponding Nash equilibrium also achieves social optimality. 
We also provide sufficient conditions for the learning dynamics to converge to a fixed point so that the adaptive incentive mechanism eventually induces a socially optimal outcome.  
Finally, we demonstrate that the sufficient conditions for convergence are satisfied in a variety of games, including (i) atomic networked quadratic aggregative games, (ii) atomic Cournot competition, and (iii) non-atomic network routing games.  
\end{abstract}

\section{Introduction}
The design of incentive mechanisms plays a crucial role in many social-scale systems, where the system outcomes depend on the selfish behavior of a large number of interacting players (human users, service providers, and operators). The outcome arising from such strategic interaction -- Nash equilibrium -- often leads to suboptimal societal outcome. This is due to the fact that individual players often ignore the externality of their actions (i.e. how their actions affect the cost of others) when minimizing their own cost. An important way to address the issue of externality is to provide players with incentives that align their individual goal of cost minimization with the goal of minimizing the total cost of the society (\cite{paccagnan2019incentivizing,levin1985taxation,ho1982control,bacsar1984affine}).

The problem of incentive design is {further} complicated when the design faces a set of learning agents who are repeatedly updating their strategies (\cite{barrera2014dynamic,como2021distributed,maheshwari2021dynamic}). Such a problem is particularly relevant when the physical system has experienced a random shock, and players are in the process of reaching a new equilibrium. Designing a socially optimal incentive mechanism directly based on the convergent strategy of the learning agents is challenging because such an equilibrium is typically difficult to compute in large-scale systems. The question that then arises is: how can a social planner design an adaptive incentive mechanism to influence players' learning dynamics so that the strategy learning under the adaptive mechanism leads to a socially beneficial outcome in the long run?

We propose a discrete-time learning dynamics that jointly captures the players' strategy updates and the designer's updates of incentive mechanisms. Our learning dynamics can be used for both atomic games and non-atomic games. The incentive mechanism designed by the social planner sets a payment (tax or subsidy) for each player that is added to their cost function in the game. In each time step, players update their strategies based on the opponents’ strategies and the incentive mechanism in the current step, and the social planner updates the incentive mechanism in response to players’ current strategies. We assume that the incentive update proceeds at a slower timescale than the strategy update of players. The slower evolution of incentives is in-fact a desirable characteristic for any societal scale system, where frequent changes of incentives may lead to instability in the system and may hamper participation by players. The slow evolution of incentives allows players to consider the incentives as static while updating their strategies. 

A key feature of our learning dynamics is that the incentive update in each time step is based on the externality created by each player with their current strategy. In particular, given any strategy profile, the externality of each player is evaluated as the difference between the marginal cost of their strategy on themselves and the marginal social cost. In a static incentive design problem, when all players are charged with their externality, the change of their total cost – original cost in game plus the payment – with respect to their strategy becomes identical to the change of social cost. Consequently, the induced Nash equilibrium is also socially optimal \cite{varian1994solution,davis1962externalities,pigou2017economics}. In our learning dynamics, the social planner accounts for the externality of each player evaluated at their current strategy, which evolves with players’ strategy updates. 

The externality-based incentive updates distinguish our adaptive incentive design from other recent studies on incentive mechanisms with learning agents.  The paper \cite{ratliff2020adaptive} studies the problem of incentive design while learning the cost functions of players. The authors assume that both the cost functions and incentive policies are linearly parameterized, and the incentive updates rely on the knowledge of players’ strategy update rules instead of just the current strategy as in our setting. Additionally, the paper \cite{liu2021inducing} considers a two-timescale discrete-time learning dynamics, where players adopt a mirror descent-based strategy update, and the social planner updates an incentive parameter according to a gradient descent method. The convergence of such gradient-based learning dynamics relies on the assumption that the social cost given players’ equilibrium strategy is convex in the incentive parameter. However, the convexity assumption can be restrictive since the equilibrium strategy as a function of the incentive parameter is nonconvex even in simple games.

We show that our externality-based incentive updates ensure that any fixed point of our learning dynamics corresponds to a optimal incentive mechanism, such that the induced Nash equilibrium of the game is also socially optimal (Proposition \ref{prop: Alignment}). This result is built on the fact that at any fixed point of our learning dynamics, the strategy profile is a Nash equilibrium corresponding to the incentive mechanism, and each player's payment equals to the externality created by their equilibrium strategy. Therefore, the equilibrium strategy associated with this externality-based payment also minimizes the social cost. Additionally, we present the sufficient conditions on the game such that the fixed point set is a singleton set, and thus the socially optimal incentive mechanism is unique (Proposition \ref{prop: UniquenessOfPInitFin}).

Furthermore, we provide sufficient conditions on games that guarantee the convergence of strategies and incentives induced by our learning dynamics (Theorem \ref{thm: ConvergenceFin}). Since the convergent strategy profile and incentive mechanism corresponds to a fixed point that is also socially optimal, these sufficient conditions guarantee that the adaptive incentive mechanism eventually induces a socially optimal outcome in the long run. 

In the proof of our convergence theorem, we exploit the timescale separation between the strategy update and the incentive updates. We use tools from the theory of two-timescale dynamical systems \cite{borkar1997stochastic} to analyze the convergence of strategy updates and incentive updates separately after accounting of time separation. In particular, the convergence of strategy updates can be derived from the rich literature of learning in games (\cite{fudenberg1998theory},\cite{sandholm2010population},\cite{monderer1996fictitious}, etc.) since the incentive mechanism can be viewed as static in the strategy updates thanks to the time separation. On the other hand, the convergence of incentive vectors can be analyzed via the associated continuous-time dynamical system, in which the value of the externality function is evaluated at the converged value of fast strategy update, which is the Nash equilibrium.
Our sufficient conditions are based broadly on two main techniques of proving global stability of non-linear dynamical system: \textit{(i)} cooperative dynamical systems theory \cite{hirsch1985systems} and \textit{(ii)} Lyapunov based methods \cite{sastry2013nonlinear}.

Finally, we apply our general results to three classes of games: (i) atomic networked quadratic games; (ii) atomic cournot competition; (iii) nonatomic routing games. In each class of games, we present the adaptive incentive design based on the externality of players’ strategies. We also provide sufficient conditions on the game parameters and social cost functions under which the adaptive incentive design eventually induces a socially optimal outcome.

The article is organized as follows: in Sec. \ref{sec: Model} we describe the setup of both atomic and non-atomic game considered here. In addition, we also provide the joint strategy and incentive update considered in this paper. We present the main results in Sec. \ref{sec: Results} and the applications of those results in three class of games in Sec. \ref{sec: Applications}. We conclude our work in Sec. \ref{sec: Conclusion}. 

\subsection*{Notations}
For any vector \(x\in \R^{n}\), we use $x_j$ or \(x^{j}\) to denote the \(j-\)th component of that vector. 
Given a function \(f:\R^{n}\rightarrow \R\), we use \(D_{x_i}f(x)\) to denote $\frac{\partial f}{\partial x_i}(x)$, the derivative of $f$ with respect to $x_i$ for any \(i\in\{1,2,...,n\}\). For any matrix \(A\in \R^{n\times n}\) we denote the set of eigenvalues of $A$ by \(\spec(A)\). For any set \(A\) we use \(\textsf{conv}(A)\) to denote the convex hull of the set. We use \(k\) to denote the discrete-time index and \(t\) to denote the continuous-time index.
\section{Model}\label{sec: Model}
We introduce both atomic and non-atomic static games in Sec.\ref{subsec:static}. In Sec. \ref{ssec: DynamicsModel}, we present the two-timescale dynamics of strategy learning and incentive design. 
\subsection{Static games}\label{subsec:static}
  
\subsubsection{Atomic Games}
Consider a game $\finGame$ with a finite set of players $\playerSet$. The strategy of each player $i \in \playerSet$ is $\strategyFin_i \in \strategySetFin_i$, where $\strategySetFin_i$ is a non-empty and closed interval in $\mathbb{R}$.
The strategy profile of all players is $\strategyFin=(\strategyFin_i)_{i \in \playerSet}$, and the set of all strategy profiles is $\strategySetFin \defas \prod_{i \in \playerSet} \strategySetFin_i$. 
The cost function of each player $i \in \playerSet$ is $\ell_i: \strategySetFin \to \mathbb{R}$.\footnote{We measure the outcome of our games by costs instead of utilities. Equivalently, the utility of each player is the negative value of the cost.} For any $\strategyFin_{-i} = (\strategyFin_j)_{j \in \playerSet \setminus \{i\}}$, we assume that the cost function \(\ell_i(\strategyFin_i, \strategyFin_{-i})\) is  twice-continuously differentiable and strictly convex in $\strategyFin_i$ for all $i \in \playerSet$.

A \emph{social planner} designs incentives by setting a payment $p_i x_i$ for each player $i$ that is linear in their strategy $x_i$.\footnote{Considering a linear payment is sufficient to ensure optimal incentive design in atomic games.} Here, $\incentiveFin_i$ represents the marginal payment for every unit increase in strategy of player $i$. The value of $\incentiveFin_i$ can either be negative or positive, which represents a marginal subsidy or a marginal tax, respectively.


Given the incentive vector $p=(p_i)_{i \in I}$, the total cost of each player \(i\in\playerSet\) is:
\begin{align}\label{eq: TotCost}
    \costFin_i(\strategyFin,\incentiveFin) = \lossFin_i(\strategyFin)+ \incentiveFin_i\strategyFin_i,\quad \forall \ \strategyFin\in\strategySetFin.
\end{align}



A strategy profile $\xEqFin{}(p) \in \strategySetFin$ is a
\emph{Nash equilibrium} in the atomic game $\finGame$ with the incentive vector $\incentiveFin$ if  \begin{align*}
    \costFin_i(\xEqFin{i}(p),\xEqFin{-i}(p),\incentiveFin) \leq \costFin_i(\strategyFin_i,\xEqFin{-i}(p),\incentiveFin), ~\forall \ \strategyFin_i \in \strategySetFin_i, ~\forall i \in \playerSet. 
\end{align*}
Recall that the cost $\ell_i(\strategyFin_i, \strategyFin_{-i})$ is a continuous function, and is strictly convex in $\strategyFin_i$. Additionally, the strategy set $\strategySetFin_i$ is convex for each player $i$. Therefore, we know that Nash equilibrium must exist and must be unique in $\finGame$. Moreover, we can equivalently represent a Nash equilibrium $\xEqFin{}$ as a strategy profile that satisfies the following variational inequality (\cite{facchinei2007finite}):
\begin{align}\label{eq: FiniteGameVI}
    \langle \JacobianIncentiveFin(\xEqFin{}(p),\incentiveFin) , \strategyFin-\xEqFin{}(p) \rangle \geq 0, \quad \forall \ \strategyFin\in \strategySetFin, 
\end{align}
where $\JacobianIncentiveFin(\xEqFin{}(p),\incentiveFin)= (\JacobianIncentiveFin_i(\xEqFin{}(p),\incentiveFin))_{i \in \playerSet}$, and 
 \begin{align}\label{eq: GameJacobian}
      \JacobianIncentiveFin_i(\xEqFin{}(p),\incentiveFin) = \Der_{\strategyFin_i} \costFin_i(\xEqFin{}(p),\incentiveFin) = \Der_{\strategyFin_i} \ell_i(\xEqFin{})+ \incentiveFin_i.  
    \end{align}

Furthermore, a strategy profile $\socOptFin \in \strategySetFin$ is \emph{socially optimal} if $\socOptFin$ minimizes the social cost function \(\socCostFin:\strategySetFin\to\R\). We assume that the social cost function \(\socCostFin(\strategyFin)\) is strictly convex and twice continuously differentiable in $\strategyFin$. Then, the optimal strategy profile $\socOptFin$ is unique. Additionally, from the first order conditions of optimality, we know that \(\socOptFin\) minimizes the social cost function $\socCostFin$ if and only if:
\begin{align}\label{eq: SocCostFiniteVI}
    \langle \nabla \socCostFin(\socOptFin) , \strategyFin-\socOptFin \rangle \geq 0, \quad \forall \ \strategyFin\in \strategySetFin.
\end{align}

Finally, given a strategy profile \(\strategyFin\in\strategySetFin\), we define the \emph{externality} caused by player $i$ as the difference between the marginal social cost, and the marginal cost of player $i$  with respect to $\strategyFin_i$. That is,
\begin{align}\label{eq: ExterFin}
    \externality_i(\strategyFin) = \Der_{\strategyFin_i} \socCostFin(\strategyFin) -  \Der_{\strategyFin_i} \ell_i(\strategyFin). 
\end{align}

\subsubsection{Non-atomic Games}\label{ssec: PopGame} 
Consider a game $\popGame$ with a finite set of player populations $\pop$. Each population $i \in \pop$ is comprised of a continuum set of players with mass $\massPop_i >0$. Individual players in each population can choose an action in a finite set $\stratPop_i$. The strategy of population $i \in \pop$ is $\strategyPop_i = \left(\strategyPop_i^j\right)_{j \in \stratPop_i}$, where $\strategyPop_i^j$ is the fraction of individuals in population $i$ who choose action $j \in \stratPop_i$. Then, the strategy set of population $i$ is $\strategySetPop_i =\left\{\strategyPop_i | \sum_{j \in \stratPop_i} \strategyPop_i^j = \massPop_i, ~ \strategyPop_i^j \geq 0, \forall j \in \stratPop_i\right\}$. The strategy profile of all populations is $\strategyPop=(\strategyPop_i)_{i \in \pop} \in \strategySetPop = \prod_{i \in \pop} \strategySetPop_i$. Given a strategy profile $\strategyPop\in \strategySetPop$, the cost of players in population $i \in \pop$ for choosing action $j \in \stratPop_i$ is $\lossPop_{i}^{j}(\strategyPop)$ which is assumed to be continuously differentiable. We denote $\lossPop_i(\strategyPop) = (\lossPop_i^j(\strategyPop))_{j \in \stratPop_i}$ as the vector of costs for each population $i \in \pop$.

Given any $\strategyPop \in \strategySetPop$, a social planner designs incentives by setting a payment $\incentivePop{i}{j}$ for players in population $i$ who choose action $j$. Consequently, given the incentive vector \(\incentivePop{}{} = \lr{\incentivePop{i}{j}}_{j\in S_i, i\in\pop}\), the total cost of players in each population $i \in \pop$ for choosing action $j \in \stratPop_i$ is given by: 
\begin{align}\label{eq: TotCostPop}
    \costPop{i}{j}(\strategyPop,\incentivePop{}{}) = \lossPop_i^{j}(\strategyPop)+\incentivePop{i}{j} \quad \forall \  \strategyPop\in \strategySetPop.
\end{align}

A strategy profile $\xEqPop{}{}(\ptilde) \in \strategySetPop$ is a Nash equilibrium in the nonatomic game $\popGame$ with $\incentivePop{}{}$ if 
\begin{align*}
    \forall i \in \pop, ~ \forall j \in \stratPop_i, ~ &\xEqPop{i}{j}(\ptilde) > 0, \quad  \Rightarrow \\
    &\costPop{i}{j}(\xEqPop{}{}(\ptilde),\incentivePop{}{})\leq \costPop{i}{j'}(\xEqPop{}{}(\ptilde),\incentivePop{}{}) , \quad \forall j'\in \stratPop_i. 
\end{align*}
Similar to that in atomic games, we can equivalently represent the Nash equilibrium $\xEqPop{}{}(\ptilde)$ in non-atomic game $\popGame$ as a strategy profile that satisfies the following variational inequality (\cite{sandholm2010population}): 
\begin{align}\label{eq: PopulationGameVI}
    \langle \costPop{}{}(\xEqPop{}{}(\ptilde),\incentivePop{}{}), \strategyPop-\xEqPop{}{}(\ptilde) \rangle \geq 0 \quad \forall \ \strategyPop\in \strategySetPop, 
\end{align}
where $\costPop{}{}(\xEqPop{}{}(\ptilde),\incentivePop{}{})= (\costPop{i}{}(\xEqPop{}{}(\ptilde),\incentivePop{}{}))_{i \in \pop}$. 

Note that Nash equilibrium always exists in a population game $\popGame$ \cite[Theorem 2.1.1]{sandholm2010population}. Under the assumption that the cost function $\costPop{}{}(\strategyPop, \incentivePop{}{})$ is strictly monotone in $\strategyPop$ (Assumption \ref{assm: MonotonicCostPop}), Nash equilibrium $\xEqPop{}{}$ is also unique \cite{sandholm2010population}.
\begin{assm}\label{assm: MonotonicCostPop}
For every incentive vector \(\incentivePop{}{}\), 
\begin{align*}
    \lara{ \costPop{}{}(x,\incentivePop{}{}) - \costPop{}{}(x',\incentivePop{}{}),x-x' } > 0, \quad \forall x\neq x'\in \strategySetPop.
\end{align*}
\end{assm}

Analogous to the atomic games, a strategy profile $\socOptPop \in \strategySetPop$ is socially optimal if $\socOptPop$ minimizes a social cost function $\socCostPop: \strategySetPop \to \mathbb{R}$. We assume that $\socCostPop(\strategyPop)$ is strictly convex, and twice continuously differentiable in $\strategyPop$. Therefore, $\socOptPop$ is unique, and satisfies the following variational inequality constraints: 
\begin{align}\label{eq: SocCostPopVI}
    \langle \nabla \socCostPop(\socOptPop) , \strategyPop-\socOptPop \rangle \geq 0, \quad \forall \ \strategyPop\in \strategySetPop.
\end{align}

Finally, give any $\strategyPop \in \strategySetPop$, we define the externality caused by players in population $i$ who play action $j \in \stratPop_i$ as the difference between the marginal social cost, and the cost experienced by the players in population $i$ who chooses action $j$, i.e. 
\begin{align}\label{eq: ExterPop}
    \externalityPop{i}{j}(\strategyPop) = \Der_{\strategyPop_i^{j}} \socCostPop(\strategyPop) - \lossPop_{i}^{j}(\strategyPop).
\end{align}

\subsection{Learning dynamics}\label{ssec: DynamicsModel}
We now introduce the discrete-time learning dynamics considered in this paper. For every time step \(k=1,2,...\), we denote the strategy profile in the atomic game $\G$ (resp. non-atomic game $\Gtilde$) as $\xk{k}= (\xki{k})_{i \in \playerSet}$ (resp. $\xtildek{k}= (\xtildeki{k})_{i \in \pop}$), where $\xki{k}$ (resp. $\xtildeki{k}$) is the strategy of player $i$ (population $i$) in step $k$. Additionally, we denote the incentive vector as $\pk{k}= (\pki{k})_{i \in \playerSet}$ (resp. $\ptildek{k} = (\ptildeki{k}^j)_{j \in S_i, i \in \pop}$). The strategy updates and the incentive updates are presented below: 

\medskip 
\noindent\textbf{Strategy update.} In each step $k+1$, the updated strategy is a linear combination of the previous strategy in stage $k$ (i.e. $\xk{k}$ in $\G$ and $\xtildek{k}$ in $\Gtilde$), and a new strategy (i.e. $\f(\xk{k},\pk{k}) \in \X$ in $\G$ and $\ftilde(\xtildek{k},\ptildek{k}) \in \Xtilde$ in $\Gtilde$) that depends on the previous strategy and the incentive vector in stage $k$. The relative weight in the linear combination is determined by the step-size $\stepx{k} \in (0, 1)$. 
\begin{align}
   \xk{k+1}&= (1-\stepx{k})\xk{k}+\stepx{k}\f(\xk{k},\pk{k})\tag{$x$-update}\label{eq: FinUpdateX}\\
    \xtildek{k+1} &= (1-\stepx{k})\xtildek{k}+\stepx{k}\tilde{f}(\xtildek{k},\ptildek{k})\tag{$\xtilde$-update}\label{eq: PopUpdateX}
\end{align}

We consider generic strategy updates \eqref{eq: FinUpdateX} and \eqref{eq: PopUpdateX} such that the new strategy profile $\f(\xk{k},\pk{k}) = (\f_i(\xk{k}, \pk{k}))_{i \in \playerSet}$ and  $\ftilde(\xtildek{k},\ptildek{k}) =(\ftilde_i(\xtildek{k},\ptildek{k}))_{i \in \pop}$ can incorporate a variety of strategy update rules. Two simple examples of such updates include: 
\begin{enumerate}
    \item \textit{Equilibrium update:}
        The strategy update incorporates a Nash equilibrium strategy profile with respect to the incentive vector in stage $k$. That is,
        \(\f(\strategyFin_k,\incentiveFin_k) = \xEqFin{}(\incentiveFin_k)\) and \(\tilde{f}(\strategyPop_k,\incentivePop{k}{})=\xEqPop{}{}(\incentivePop{k}{})\).
        \item \textit{Best response update:} The strategy update incorporates a best response strategy with respect to the strategy and the incentive vector in the previous step, i.e. \(\f_i(\xk{k},\pk{k}) = \mathrm{BR}_i(x_{k},\pk{k}) = \underset{y_i\in\strategySetFin_i}{\arg\min} \ \costFin_i(y_i,x_{-i,k},\pk{k}),\)
        \(\tilde{f}_i(\xtilde_{k},\ptildek{k}) = \tilde{\mathrm{BR}}_i(\tilde{x}_{k},\ptildek{k}) = \underset{\tilde{y}_i\in\strategySetPop_i}{\arg\min}  \  \tilde{y}_i^\top \costPop{i}{}(\tilde{x}_{k},\ptildek{k}).\)
\end{enumerate}

\medskip 
\noindent\textbf{Incentive update.} In each step $k+1$, the updated incentive vector is a linear combination of the previous vector in step $k$ (i.e. $\pk{k}$ in $\G$ and $\ptildek{k}$ in $\Gtilde$), and the externality (i.e. $\mdFin(\xk{k})$ in $\G$ and $\mdPop(\xtildek{k})$ in $\Gtilde$) based on the strategy profile in step $k$. The relative weight in the linear combination is determined by the step size $\stepp{k} \in (0, 1)$. 
\begin{align}
\pk{k+1} &= (1-\stepp{k})\pk{k}+\stepp{k} \mdFin(\xk{k}); \tag{$p$-update}\label{eq: FinUpdateP}\\
   \ptildek{k+1} &= (1-\stepp{k})\ptildek{k}+\stepp{k} \mdPop(\xtildek{k}); \tag{$\ptilde$-update}\label{eq: PopUpdateP}
\end{align}

The incentive updates \eqref{eq: FinUpdateP}-\eqref{eq: PopUpdateP} modify the incentives on the basis of the externality caused by the players. We emphasize that this update is adaptive to the evolution of players' strategies since the externality is evaluated based on players' current strategies. Moreover, the computation of each player's externality only requires that the social planner knows the gradients of its own costs and those of the players, evaluated at the players' current strategy profile.


The joint evolution of strategy profiles and incentive vectors $(\xk{k}, \pk{k})_{k=1}^{\infty}$ (resp. $(\xtildek{k}, \ptildek{k})_{k=1}^{\infty}$) in the atomic game $\G$ (resp. non-atomic game $\Gtilde$) is governed by the learning dynamics \eqref{eq: FinUpdateX} -- \eqref{eq: FinUpdateP} (resp. \eqref{eq: PopUpdateX} -- \eqref{eq: PopUpdateP}). The step-sizes $(\stepx{k})_{k=1}^{\infty}$ and $(\stepp{k})_{k=1}^{\infty}$ determine the speed of strategy updates and incentive updates. We make the following assumption on step-sizes: 

\begin{assm}\label{assm: StepSizeAssumption}\hfill 
\begin{itemize}
\item[(i)] $\sum_{k=1}^{\infty}\stepx{k}=\sum_{k=1}^{\infty}\stepp{k}=+\infty$, $ \sum_{k=1}^{\infty}\stepx{k}^2+\stepp{k}^2 < +\infty$. 
\item[(ii)]$\lim_{k\to\infty}\frac{\stepp{k}}{\stepx{k}}=0$.
\end{itemize}
\end{assm}

\medskip 
In Assumption \ref{assm: StepSizeAssumption}, \emph{(i)} is a standard assumption on step-sizes that allow us to analyze the convergence of the discrete-time learning dynamics. Additionally, \emph{(ii)} assumes that the incentive update occurs at a slower timescale compared to the update of strategies.

Since the assumption on stepsizes (Assumption \ref{assm: StepSizeAssumption} \emph{(ii)}) ensures that the incentive evolves on a slower timescale than the strategies, players may view the incentive mechanism as approximately static (although not completely fixed) when updating their strategies. One can show that with any fixed incentive mechanism, strategy updates with Nash equilibrium being the new strategy always converges. On the other hand, although best response updates, which we also consider, do not converge in all games, they converge in many practically-relevant games such as zero sum games \cite{harris1998rate}, potential games \cite{swenson2018best}, and dominance solvable games \cite{nachbar1990evolutionary}. Additionally, our strategy updates \eqref{eq: FinUpdateX} and \eqref{eq: PopUpdateX} can incorporate many other learning dynamics; their convergence properties in static game environments have been extensively studied in the literature, in both atomic and nonatomic games \cite{mazumdar2020gradient}, \cite{sandholm2010population}, \cite{fudenberg1998theory}, \cite{monderer1996fictitious}.

We emphasize that the convergence of strategy updates with fixed incentive mechanism is not the focus of our paper. Instead, our goal is to characterize conditions under which the adaptive incentive updates \eqref{eq: FinUpdateP} and \eqref{eq: PopUpdateP} converge to a socially optimal mechanism. We note that such convergence cannot be achieved in scenarios where the strategy updates do not converge even with completely fixed incentive vector. Therefore, we impose the following assumption that the strategy updates consider in our dynamics converge to a Nash equilibrium with any fixed incentive vector.


\begin{assm}\label{assm: ConvergenceStrategy}
 In \(\finGame\) (resp. \(\popGame\)), the
 updates \eqref{eq: FinUpdateX} (resp. \eqref{eq: PopUpdateX}) starting from any
 initial strategy $\strategyFin_1$ (resp. \(\strategyPop_1\)) with \((\incentiveFin_k) \equiv \incentiveFin\) for any \(\incentiveFin\) (resp. \((\incentivePop{k}{}) \equiv \incentivePop{}{}\) for any \(\incentivePop{}{}\)),  satisfies
 \(\lim_{k\to\infty}\strategyFin_k=\xEqFin{}(\incentiveFin)\) (resp. \(\lim_{k\to\infty}\strategyPop_k=\xEqPop{}{}(\incentivePop{}{})\)), where $\xEqFin{}(\incentiveFin)$ (resp. $\xEqPop{}{}(\incentivePop{}{})$) is the Nash equilibrium corresponding to \(\incentiveFin\) (resp. \(\incentivePop{}{}\)). 
\end{assm}


\section{General results}\label{sec: Results}

In Sec \ref{subsec:fpa} we characterize the set of fixed points of the dynamic updates \eqref{eq: FinUpdateX}-\eqref{eq: FinUpdateP} and \eqref{eq: PopUpdateX}-\eqref{eq: PopUpdateP}, and show that any fixed point corresponds to a socially optimal incentive mechanism such that the induced Nash equilibrium strategy profile minimizes the social cost. In Sec. \ref{subsec: Convergence}, we provide a set of sufficient conditions that guarantee the convergence of strategies and incentives in our learning dynamics. Under these conditions, our learning dynamics designs an adaptive incentive mechanism that eventually induces a socially optimal outcome.

\subsection{Fixed point analysis}\label{subsec:fpa}
We first characterize the set of fixed points of our learning dynamics \eqref{eq: FinUpdateX}-\eqref{eq: FinUpdateP}, and \eqref{eq: PopUpdateX}-\eqref{eq: PopUpdateP}) as follows: 
{\small \begin{subequations}
\begin{align}
   &\text{Atomic game $\G$,} ~ \left\{(x, p) | f(x, p) = x, ~ e(x) = p\right\}, \label{subeq:fixedfin}\\
   &\text{Nonatomic game $\Gtilde$,} ~ \left\{(\tilde{x}, \tilde{p}) | \tilde{f}(\tilde{x}, \tilde{p}) = \tilde{x}, ~ \tilde{e}(\tilde{x}) = \tilde{p}\right\}.\label{subeq:fixedpop}
\end{align}
\end{subequations}}

We can check that if the learning dynamics starts with a fixed point strategy and incentive vector, then the strategies and incentive vectors remain at that fixed point for all time steps. Moreover, under Assumption \ref{assm: ConvergenceStrategy}, we know  that for any incentive vector $p$ (resp. $\tilde{p}$), a strategy profile that satisfies $f(x, p)=x$ (resp. $\tilde{f}(\tilde{x}, \tilde{p}) = \tilde{x}$) in game $\G$ (resp. $\Gtilde$) must be a Nash equilibrium $\xEqFin{}(p)$ (resp. $\xEqPop{}{}(p)$). Thus, from \eqref{subeq:fixedfin} -- \eqref{subeq:fixedpop}, we can write the set of incentive vectors at the fixed point as follows: 
\begin{subequations}
\begin{align*}
   \text{Atomic game $\G$,} ~ \Peq &= \{(\peq_i)_{i\in\playerSet} | \exterFin(\xEqFin{}(\peq)) = \peq \}, \\
    \text{Nonatomic game $\Gtilde$,} ~ \Ptildeeq &= \{(\ptildeeq_i)_{i\in\playerSet} | \exterPop (\tilde{x}^{*}(\ptildeeq)) = \ptildeeq \}.
\end{align*}
\end{subequations}
That is, at any fixed point, the incentive of each player is set to be equal to the externality evaluated at their equilibrium strategy profile. 

Our next proposition shows that the fixed point set $\Peq$ (resp. $\Ptildeeq$) is non-empty in $\G$ (resp. $\Gtilde$). Moreover, given any fixed point incentive parameter $\peq \in \Peq$ and $\ptildeeq \in \Ptildeeq$, the corresponding Nash equilibrium is socially optimal. 


\begin{prop}\label{prop: Alignment}
In $\G$ (resp. $\Gtilde$), the set $\Peq$ (resp. $\Ptildeeq$) is non-empty. Additionally, any $\peq \in \Peq$ (resp. $\ptildeeq \in \Ptildeeq$) is socially optimal in that $\xEqFin{}(\peq) = \socOptFin$ (resp. $\xEqPop{}{}(\ptildeeq)= \socOptPop$). 
\end{prop}

This result is especially interesting from perspective of implementation because the existence of the optimal incentives implies that for \(\finGame\) there exists a \emph{linear} incentive policy (as in \eqref{eq: TotCost}) which is optimal. Moreover for \(\popGame\) there exists a \emph{constant} incentive policy (as in \eqref{eq: TotCostPop}) that is optimal.

Proof of Proposition \ref{prop: Alignment} is based on Brouwer's fixed point theorem. The boundedness of the strategy space allows us to construct convex compact sets which maps to itself under \(\exterFin(\xEqFin{}(\cdot))\) (resp. \(\exterPop(\xEqPop{}{}(\cdot))\)) in \(\finGame\) (resp. in \(\popGame\)). 

Next, we provide sufficient conditions under which the fixed point set \(\Peq\) and \(\Ptildeeq\) are singleton. 
\begin{prop}\label{prop: UniquenessOfPInitFin}
In atomic game $\G$, the set $\Peq$ is singleton if any one of the following conditions holds: 
\begin{itemize}
    \item[(i)] The equilibrium strategy profile \(\xEqFin{}(\incentiveFin)\) is in the interior of the strategy set $\strategySetFin$ for any $p$ \item[(ii)] $\lara{\mdFin(x)-\mdFin(x'),x-x'} > 0$ for all $x \neq x'$
\end{itemize}
In non-atomic game $\Gtilde$, $\Ptildeeq$ is singleton if the externality function $\exterPop(\cdot)$ satisfies  Assumption \ref{assm: MonotonicCostPop} and condition (ii). 

\end{prop}

Under the sufficient condition in Proposition \ref{prop: UniquenessOfPInitFin}, in \(\finGame\) (resp. \(\popGame\)) there exists a unique incentive mechanism in \(\Peq\) (resp. \(\Ptildeeq\)) such that players pay for their externality at equilibrium. From Proposition \ref{prop: Alignment}, such a mechanism induces a socially optimal outcome.


\subsection{Convergence to optimal incentive mechanism}\label{subsec: Convergence}
The next result provides sufficient conditions for strategies and incentives updates \eqref{eq: FinUpdateX}-\eqref{eq: FinUpdateP} and \eqref{eq: PopUpdateX}-\eqref{eq: PopUpdateP} to converge to social optimality. 

\begin{theorem}\label{thm: ConvergenceFin}
Under Assumptions \ref{assm: StepSizeAssumption} and \ref{assm: ConvergenceStrategy}, the sequence of strategies and incentives induced by the discrete-time dynamics \eqref{eq: FinUpdateX}-\eqref{eq: FinUpdateP} in $\G$ satisfies 
\begin{align}\label{eq: ConvergenceDiscrete}
   \lim_{k\to\infty}(\xk{k},\pk{k})= (\socOptFin, \peq) 
\end{align}
if at least one of the following conditions holds: 
\begin{itemize}
\item[(C1)] If $e_i(\xEqFin{}(0)) \geq 0$, then $\lim_{p \to \infty} e_i(\xEqFin{}(p)) - p_i <0$ for all $i\in\playerSet$. If $e_i(\xEqFin{}(0)) \leq 0$, then $\lim_{p \to -\infty} e_i(\xEqFin{}(p)) - p_i >0$ for all $i\in\playerSet$.\footnote{$p \to \infty$ means $p_i \to \infty$ for all $i$.} Moreover, \(\frac{\partial \mdFin_i(\xEqFin{}(\incentiveFin)) }{\partial \incentiveFin_j} > 0\) for all \(\incentiveFin\in\R^{\actDimFin}\) and all $i \neq j$. 
    

\item[(C2)] There exists a continuously differentiable, positive definite and decrescent function \footnote{A function $V: \mathbb{R}^n \rightarrow \mathbb{R}$ is positive definite if $V(x) \geq \alpha_1(\|x\|)$ for some  continuous, strictly increasing function $\alpha_1(\cdot)$ such that $\alpha_1(0) =0$, and $\alpha_1(t) \rightarrow \infty$ as $t \rightarrow \infty$. $V$ is decrescent if $V(x) \leq \alpha_2(\|x\|)$ for some  continuous, strictly increasing function $\alpha_2(\cdot)$ such that $\alpha_2(0) =0$.} $V(p):\mathbb{R}^{\actDimFin} \rightarrow \mathbb{R}_{+}$ such that \(V(\pEqFin{}) = 0\) and \(V(p)>0\) for all \(p\neq \pEqFin{}\). Moreover:
\[
\nabla V(p)^\top \lr{e(\xEqFin{}(p))-p} < -\omega(\|p-\pEqFin{}\|) \quad \forall \ p\neq \pEqFin{}, 
\]
where \(\omega(\cdot)\) is strictly increasing, and satisfies \(\omega(0)=0\). 
\end{itemize}
Analogously, the sequence of strategies and incentives in $\Gtilde$ induced by \eqref{eq: PopUpdateX} and \eqref{eq: PopUpdateP} satisfies $\lim_{k\to\infty}(\xtildek{k},\ptildek{k})= (\socOptPop, \ptildeeq)$ if the externality function $\tilde{e}$ satisfies at least one of (C1) and (C2). 
\end{theorem}

Owing to Assumption \ref{assm: StepSizeAssumption}, we utilize the timescale separation between the strategy update \eqref{eq: FinUpdateX} and the incentive update \eqref{eq: FinUpdateP} to prove Theorem \ref{thm: ConvergenceFin}. Indeed, the two-timescale stochastic approximation theory \cite{borkar1997stochastic} suggests that the strategy update \eqref{eq: FinUpdateX} is a fast transient while the incentive update \eqref{eq: FinUpdateP} is a slow component. Therefore while considering the fast strategy update one should expect that slow incentive updates are quasi-static. Consequently, Assumption \ref{assm: ConvergenceStrategy} in game \(\finGame\) along with Assumption \ref{assm: StepSizeAssumption} ensures that the tuple \((\strategyFin_k,\incentiveFin_k)\) converges to the set \(\{(\xEqFin{}(\incentiveFin),\incentiveFin):\incentiveFin\in\R^{|\playerSet|}\}\) \cite{borkar1997stochastic}. Thus for sufficiently large values of \(k\), the update \(\strategyFin_k\) closely tracks \( \xEqFin{}(\incentiveFin_{k})\). Therefore, we consider the following update to analyze the convergence of the slow incentive update \eqref{eq: FinUpdateP}:
\begin{align}\label{eq: SlowDisFin}
    \incentiveFinSep_{k+1} &= (1-\stepp{k})\incentiveFinSep_{k}+\stepp{k} \mdFin(\xEqFin{}(\incentiveFinSep_k)).
\end{align}
Since the step sizes \(\{\stepP_k\}\) are asymptotically going to zero and has infinite travel (Assumption \ref{assm: StepSizeAssumption}-(i)) we can approximate the updates in \eqref{eq: SlowDisFin} by the following continuous-time dynamical system:
\begin{align}
    \dot{\incentiveFinSep}(t) &=  \exterFin(\xEqFin{}(\incentiveFinSep(t))) - \incentiveFinSep(t), \label{subeq:p_fin1}
\end{align}
Convergence of discrete-time updates \eqref{eq: FinUpdateX}-\eqref{eq: FinUpdateP} then hold if the flow of \eqref{subeq:p_fin1} globally converges to \(\Peq\). 

Requirements (C1) in Theorem \ref{thm: ConvergenceFin} are sufficient conditions for convergence of the trajectories of \eqref{subeq:p_fin1} to the set \(\Peq\). This condition is based on cooperative dynamical systems theory \cite{hirsch1985systems}.  Intuitively, condition (C1) demands that in the equilibrium if the players inflict (resp. alleviate) some externality when no incentive are applied then there should exist high enough prices (resp. subsidies) which can compensate for the externality.  Moreover, it demands that higher prices (resp. subsidies) on other player increases the externality inflicted (resp. alleviated) by a player. 


Requirement (C2) in Theorem \ref{thm: ConvergenceFin} on the other hand ensures convergence by positing existence of a Lyapunov function \cite{sastry2013nonlinear} that is strictly positive everywhere except at \(\Peq\) and decreases along the flow of {\eqref{subeq:p_fin1}}. 


Note that either one of the conditions (C1) or (C2) guarantees  the convergence of the flow of the slow system \eqref{subeq:p_fin1} to \(\Peq\). This in addition to the convergence of the fast strategy update  (Assumption \ref{assm: ConvergenceStrategy}) leads to the  convergence of the discrete-time dynamics \eqref{eq: FinUpdateX}-\eqref{eq: FinUpdateP} \cite[Chapter 6]{borkar1997stochastic}.

Thus, we have shown that there exists an incentive which induces an equilibrium which is socially optimal and the externality based pricing update along with any strategy update, satisfying requirements of Theorem \ref{thm: ConvergenceFin}, converges to the optimal incentive. 

\section{Applications}\label{sec: Applications}
In this section, we apply our general results to three classes of games that are practically relevant: (Sec \ref{ssec: Quad}) Atomic networked aggregative games; (Sec \ref{ssec: Cournot}) Atomic Cournot games; and (Sec \ref{ssec: Routing}) Non-atomic routing games. In each case, we show that our adaptive incentive mechanism asymptotically induces a socially optimal outcome. 
\subsection{Atomic Networked Aggregative Games}\label{ssec: Quad}
We consider a finite set of players \(\playerSet\) who are connected in a network. The strategy of each player \(i\in\playerSet\) is a real number \(\strategyFin_i\in \R\). We represent the network that connects players by a matrix \(\Zmat=(\Zmat_{ij})_{i,j\in\playerSet}\), where $\Zmat_{ij}$ captures the impact of player $j$'s strategy $x_j$ on player $i$'s cost. We assume that \(\Zmat_{ii}=0\) for all $i \in \playerSet$. The cost of each player \(i\in\playerSet\) given any strategy profile \(\strategyFin=(\strategyFin_i)_{i\in\playerSet}\) is a quadratic function as follows: 
\[
\lossFin_i(\strategyFin)=\frac{1}{2}{}\strategyFin_i^2-\kI_i\strategyFin_i\agg_i(\strategyFin)
\]
where \(\kI_i>0\) and \(\agg_i(\strategyFin)=\sum_{j \in \playerSet}\Zmat_{ij}\strategyFin_j\) is the average strategy of player $i$'s neighbors weighted by the network matrix $\Zmat$. That is, $\agg_i(\strategyFin)$ captures the network effect of opponents' strategies on  player $i$.   



Networked aggregative games are applicable in a variety of settings, where players' strategies and costs are affected by those around them. Examples of such settings include peer effects, investment in networked markets, and cross-neighborhood impacts of crime \cite{jackson2015games}.

A social planner designs an incentive mechanism that charges each player $i$ with payment $p_i x_i$. The total cost of player \(i\) 
under strategy profile \(\strategyFin\) and  \(\incentiveFin\) is   
\[
\costFin_i(\strategyFin,\incentiveFin) = \frac{1}{2}{}\strategyFin_i^2 - \strategyFin_i\kI_i\agg_i(\strategyFin) + \incentiveFin_i\strategyFin_i.
\]

Let \(\Kmat=\textsf{diag}([\kI_1,\kI_2,...,\kI_{\numPlayers}]^\top)\). We assume that the matrix \((I-\Kmat\Zmat)\) is invertible and let \(L=(I-\Kmat\Zmat)^{-1}\). In the economics literature, the matrix \((I-\Kmat\Zmat)^{-1}\) is referred to as the \emph{Leontief matrix}, where the \(ij\) entry of this matrix captures how the payment of player $j$ affects the equilibrium strategy of player $i$ (\cite{parise2017sensitivity}). 

For any $p$, we show that the aggregative game has a unique Nash equilibrium given by: 
\begin{align}\label{eq: NashQuad}
\xEqFin{}(p)=-(I-\Kmat\Zmat)^{-1}p.
\end{align}


Given $x$, the cost of the social planner is  $\socCostFin(\strategyFin)=\frac{1}{2}\sum_{i\in I}
(x_i- \xi_i)^2$ for $\xi_i \in \mathbb{R}$, where $\socOptFin=(\xi_i)_{i \in I}$ is the planner's socially optimal strategy profile. Moreover, from \eqref{eq: ExterFin}, the externality caused by player $i$ given the strategy profile $x$ is \(\exterFin_i(\strategyFin)= \xi_i + \kI_i \agg_i(\strategyFin)\).


We consider the learning dynamics, where players update their strategies using best response. Given any strategy profile $x$ and incentive $p$, the best response of player \(i\) is \(\brUpdate_i(\strategyFin_{-i},\incentiveFin_i) = \kI_i\agg_i(\strategyFin)-\incentiveFin_i\) and the best response vector \(\brUpdate(x,p) = (\brUpdate_i(x_{-i},p_i))_{i\in\playerSet}\). Thus the strategy update is \(\stratUpdateFin(\strategyFin,\incentiveFin) =BR(\strategyFin, \incentiveFin)= \Kmat\Zmat \strategyFin - \incentiveFin\). Then, the discrete-time leaning dynamics \eqref{eq: FinUpdateX} -- \eqref{eq: FinUpdateP} can be written as follows: 
\begin{subequations}
\begin{align}
    \strategyFin_{k+1} &= (1-\stepX_k)\strategyFin_k + \stepX_k \lr{\Kmat\Zmat \strategyFin_k - \incentiveFin_k},\label{eq: strategyQuad}\\
    \incentiveFin_{k+1} &= (1-\stepP_k)\incentiveFin_k + \stepP_k\lr{\xi + \Kmat \Zmat\strategyFin_k},\label{eq: PriceQuad}
\end{align}
\end{subequations}
and the step-sizes  \(\{\stepX_k\}_{k=1}^{\infty},\{\stepP_k\}_{k=1}^{\infty}\) satisfy Assumption \ref{assm: StepSizeAssumption}.


We show that there exists a unique $\peq$ such that the induced equilibrium strategy profile $x^*(\peq)$ equals to the socially optimal strategy $\socOptFin$. Moreover, we also provide sufficient condition on the Leontief matrix under which our learning dynamics converges the unique socially optimal incentive mechanism. 
\begin{prop}\label{prop: QuadIncentive}
The unique socially optimal incentive mechanism is \(\pEqFin{} = (I-\Kmat\Zmat)\Bmat\). Furthermore if the real part of eigenvalues of \((\Qmat-\Kmat\Zmat)\) is positive, i.e. \(\spec(L)\subset \C_{+}^{\circ}\), then the discrete-time learning dynamics \eqref{eq: strategyQuad}-\eqref{eq: PriceQuad} satisfy \(\lim_{k\to \infty} (\strategyFinDis_k,\incentiveFinDis_k) = (\socOptFin,\pEqFin{})\). 
\end{prop}
 
From \eqref{eq: NashQuad}, we know that for any $p$, there exists a unique equilibrium strategy profile $x^*(p)$ that is linear in $p$. Then, we obtain the socially optimal incentive $\peq$ that satisfies $x^*(\peq)=\socOptFin$. Additionally, $\peq$ also satisfies that $e(\socOptFin( \peq)) = \peq$, and therefore $\peq$ is a fixed point of the learning dynamics. To show the convergence results in Proposition \ref{prop: QuadIncentive} we verify the conditions in Theorem \ref{thm: ConvergenceFin} holds. The condition \(\spec(L)\subset\C_{+}^{\circ}\) ensures that Assumption \ref{assm: ConvergenceStrategy} holds. Indeed, we show that the strategy update \eqref{eq: strategyQuad} with fixed incentives asymptotically track the flow of a continuous-time linear dynamical system. The condition \(\spec(L)\subset \C_{+}^{\circ}\) ensures that  the flow of continuous-time dynamical system asymptotically converges to fixed points of \eqref{eq: strategyQuad} with fixed incentives.
Finally, we verify that if \(\spec(L)\subset \C_{+}^{\circ}\) then \(K=-L\subset \C_{-}^{\circ}\) and there  exists a Lyapunov function that satisfies (C2) in Theorem \ref{thm: ConvergenceFin}. In particular, the Lyapunov function is given by 
\[V(\incentiveFin) = (\incentiveFin-\pEqFin{})^\top M(\incentiveFin-\pEq)\]
where $M$ is a matrix that satisfies $K^{\top}M + M K = -I$.\footnote{The existence of a matrix \(M\) is guaranteed by the Lyapunov theorem \cite{callier2012linear} as \(\spec(K)\subset \C_{-}^{\circ}\).}
\subsection{Atomic Cournot Competition}\label{ssec: Cournot}
A finite set of firms $\playerSet$ compete in a single market. The strategy of each firm \(i\in\playerSet\) is its production quantity \(x_i\). Given any strategy profile \(x=(x_i)_{i\in\playerSet}\), the price of the good is $\price(x) = \priceOne - \priceTwo \sum_{i\in\playerSet}x_i$ with $\priceOne, \priceTwo >0$. The per-unit production cost of the good is \(\prodCost\). Then, the cost function of firm $i \in \mathcal{I}$ (written as negative of the profit) is given by:
\begin{align}
    \lossFin_i(x) & = - x_i \price(x) +  \prodCost x_i
\end{align}

A social planner designs an incentive mechanism that charges each player $i$ with payment $p_i x_i$. The total cost of firm \(i\in \playerSet\) given \(x\) and \(p\) is: 
\begin{align*}
    c_i(x, p) = -x_i \price(x) + (\prodCost+p_i) x_i 
\end{align*}
The game has a unique Nash equilibrium given by: \footnote{We assume that $\theta$ is large enough such that $x^*(p) > 0$ for all $p$ in a neighborhood of the socially optimal incentive $p^{\dagger}$.}
\begin{align}\label{eq: NashCournot}
    \xEqFin{i}(p) = \frac{1}{\priceTwo(|\playerSet|+1)} \left(\priceOne-\prodCost -|\mathcal{I}| \incentiveFin_i + \sum_{j \neq i}\incentiveFin_j\right)
\end{align}

The goal of the social planner is to minimize the aggregate cost of players while also accounting for the environmental cost of good production, which is unpriced in equilibrium. We model the environmental cost to be a  quadratic function of production following \cite{cai2019role}. 
Thus, the social cost function is
\(
    \Phi(x) =  {\sum_{i=1}^{n}  \ell_i(x)}+ \lam{\sum_{i=1}^{n} x_i^2}
\)
where $\lam> 0$ is a parameter that determines the relative weight between the firm costs and environmental cost. Finally, the externality \eqref{eq: ExterFin} caused by of a firm \(i\in\playerSet\) is $e_i(x)= 2 \lambda x_i + \priceTwo\sum_{j \neq i} x_j$. 

We consider best response strategy updates. Given any $x_{-i}$, the best response of firm \(i\in\playerSet\) is:
\[
BR_i(\strategyFin_{-i},\incentiveFin_i) = \frac{\priceOne - \priceTwo \sum_{j \neq i} x_j - \nu - p_i}{2 \priceTwo}.
\] 
Following \eqref{eq: FinUpdateX} -- \eqref{eq: FinUpdateP}, we can write the updates of strategies and incentives as follows: 
\begin{subequations}
\begin{align}
    \strategyFin_{i,k+1} &= (1-\stepX_k)\strategyFin_{i,k} + \stepX_k \lr{\frac{\priceOne - \priceTwo \sum_{j \neq i} x_{j,k} - \nu - p_{i,k}}{2 \priceTwo}}\label{eq: CournotStrat}, \\
    \incentiveFin_{i,k+1} &= (1-\stepP_k)\incentiveFin_{i,k} + \stepP_k ( \priceTwo \sum_{j \neq i} \strategyFin_{j, k} + 2 \lambda \strategyFin_{i,k}).\label{eq: CournotIncentive}
\end{align}
\end{subequations}
and the step-sizes  \(\{\stepX_k\}_{k=1}^{\infty},\{\stepP_k\}_{k=1}^{\infty}\) satisfy Assumption \ref{assm: StepSizeAssumption}.

We can show that for any fixed $p$, the best response learning dynamics \eqref{eq: CournotStrat} converges to a Nash equilibrium $x^*(p)$ associated with $p$. Indeed, we show that the strategy update \eqref{eq: CournotStrat} with fixed incentives asymptotically track the flow of a continuous-time linear dynamical system whose flow asymptotically converges to the Nash equilibrium \(\xEqFin{}(p)\). Thus, Assumption \ref{assm: ConvergenceStrategy} is satisfied. 

The next proposition shows that the optimal incentive $\peq$ is unique. Moreover, the incentive vectors induced by \eqref{eq: CournotIncentive} converge to the socially optimal incentive $\peq$ if the weight of environmental cost, $\lambda$, is sufficiently high. 
\begin{prop}\label{prop: CournotConvergence}
There exists a unique socially optimal incentive mechanism $\incentiveFin^{\dagger}$ that satisfies $\incentiveFin^{\dagger} = \exterFin(\strategyFin^*(\incentiveFin^{\dagger}))$. Given $\peq$, the induced equilibrium strategy profile is socially optimal, i.e.  $\strategyFin^*(\incentiveFin^{\dagger}) = \strategyFin^{\dagger}$. Moreover, the discrete-time learning dynamics \eqref{eq: CournotStrat}-\eqref{eq: CournotIncentive} satisfy $\lim_{k \rightarrow \infty} (\strategyFin_k, \incentiveFin_k) = (\socOptFin, \pEqFin{})$ if $\lam > \priceTwo$. \end{prop}

Recall that $\lambda$ is the weight of environmental cost in the social cost function, and $\delta$ is the increase of firm cost with respect to the increase of production level. The sufficient condition $\lam > \priceTwo$ states that if the social planner assigns higher weight to the environmental cost compared to the per-unit increase of firm cost, then the adaptive incentive mechanism can asymptotically induce a socially optimal outcome. 

The proof of Proposition \ref{prop: CournotConvergence} follows similarly to that of Proposition \ref{prop: QuadIncentive}. We show that there is a unique incentive $\peq$ such that the corresponding Nash equilibrium as in \eqref{eq: CournotStrat} equals to the socially optimal strategy profile, and $\peq$ is a fixed point of the discrete-time learning dynamics \eqref{eq: CournotIncentive}. Moreover, we show that when $\lam > \priceTwo$, we can construct a Lyapnov function that satisfies (C2) in Assumption $\eqref{assm: ConvergenceStrategy}$. Therefore, following Theorem \eqref{thm: ConvergenceFin}, we can conclude that the discrete-time learning dynamics converges to a socially optimal outcome.
\subsection{Non-atomic routing games}\label{ssec: Routing}

A traveler population with total demand of 1 make routing decisions on a parallel-route network, where a single origin - destination pair is connected by a finite set of routes \(S\). The strategy of the traveler population is $\strategyPop= (\strategyPop^j)_{j \in S}$, where $\strategyPop^j$ is the mass of travelers who choose route $j \in S$. The population's strategy set is $\strategySetPop=\{\strategyPop|\sum_{j \in S} \strategyPop^j=1, ~ \strategyPop^j \geq 0, ~\forall j \in S\}$.

Given any $\strategyPop$ and any route $j \in S$, the travel time cost $\ell^j(\strategyPop^j)$ is a \emph{strictly-increasing} and \emph{convex} function of the mass of travelers who take route $j \in S$. This reflects the congestible nature of the traffic routes and the fact that the travel time increases faster as more travelers take that route. 

A social planner designs a tolling mechanism $\incentivePop{}{} = (\incentivePop{}{j})_{j \in S}$, where the toll price of route $j$ is $\incentivePop{}{j}$. Given any $\strategyPop$ and $\incentivePop{}{}$, the total cost experienced by travelers who take route $j$ is \(
      \costPop{}{j}(\strategyPop, \incentivePop{}{}) = \latency^j(\strategyPop^j) + \incentivePop{}{j}
\). 

Given any toll vector $\incentivePop{}{}$, the routing game has a unique Nash equilibrium  \(\xEqPop{}{}(\incentivePop{}{})\). The goal of the social planner is to minimize the total cost of all routes in the network, i.e.
\(
\socCostPop(\strategyPop) = \sum_{j \in S} \strategyPop^j \latency_j(\strategyPop^j).
\)
We can check that  \(\socCostPop(\strategyPop)\) is strictly convex in $\strategyPop$, and thus the socially optimal strategy $\tilde{x}^{\dagger}$ is unique.  
Finally, following from  \eqref{eq: ExterPop}, the externality caused by travelers on route  \(j\in S\) is \(
    \exterPop^j(\strategyPop)= \tilde{x}^j\frac{d\ell^j(\tilde{x}^j)}{d{\tilde{x}^j}}.
\)

We consider perturbed best response strategy updates. Given any $\strategyPop$, the perturbed best response strategy is \[\tilde{f}^i(\strategyPop, \incentivePop{}{})= \frac{\exp(-\para \costPop{}{i}(\strategyPop,\incentivePop{}{}))}{\sum_{j\in S}\exp(-\para \costPop{}{j}(\strategyPop,\incentivePop{}{}))},\]
where $\eta$ evaluates the sensitivity of travelers' route choices with respect to the costs. We note that as $\para \to \infty$, the perturbed best response strategy reduces to a best response strategy that only chooses routes with the minimal costs. 

The discrete-time learning dynamics is: 
\begin{subequations}
\begin{align}
    \strategyPop^j_{k+1} &=(1-\stepX_k)\strategyPop^j_{k} + \stepX_k \frac{\exp(-\para \costPop{}{j}(\strategyPop_k,\incentivePop{k}{}))}{\sum_{j\in S}\exp(-\para \costPop{}{j}(\strategyPop_k,\incentivePop{k}{}))}, \label{eq: stratRouting} \\
    \incentivePop{k+1}{j} &= (1-\stepP_k)\incentivePop{k}{j}+\stepP_k \strategyPop_{k}^j
    \frac{d \latency^j(\strategyPop_{k}^{j})}{d \strategyPop_k^j}. \label{eq: incentiveRouting}
\end{align}
\end{subequations}
and the step-sizes  \(\{\stepX_k\}_{k=1}^{\infty},\{\stepP_k\}_{k=1}^{\infty}\) satisfy Assumption \ref{assm: StepSizeAssumption}. Moreover, we can show that for any fixed \(\incentivePop{}{}\), the perturbed best response dynamics \eqref{eq: stratRouting} converges to the perturbed equilibrium. Indeed, due to Assumption \ref{assm: StepSizeAssumption}-i) the discrete-time updates \eqref{eq: stratRouting}, with fixed incentive \(\incentivePop{}{}\), tracks the flow of a cooperative continuous-time dynamical system \cite{hirsch1985systems} whose flows converges to the perturbed equilibrium. Thus Assumption \ref{assm: ConvergenceStrategy} holds. 

The next proposition shows that the optimal incentive $\pEqPop{}{}$ is unique. Moreover, the incentive vectors induced by \eqref{eq: incentiveRouting} converge to the socially optimal incentive $\peq$.

\begin{prop}\label{prop: RoutingOptimalIncentive}
As $\eta \to \infty$, the strategies and incentives induced by the discrete-time learning dynamics \eqref{eq: stratRouting}-\eqref{eq: incentiveRouting} converge to a unique fixed point, i.e. $
    \lim_{k\to\infty} (\strategyPop_k,\incentivePop{k}{}) = (\xEqPop{}{}(\pEqPop{}{}),\pEqPop{}{})$, where $\xEqPop{}{}(\peq)$ is a Nash equilibrium given $\peq$, and $\peq$ satisfies \(\exterPop(\xEqPop{}{}(\pEqPop{}{}))=\pEqPop{}{}\). Additionally, $\ptildeeq$ is the unique optimal incentive mechanism, and the corresponding Nash equilibrium $\xEqPop{}{}(\peq) = \tilde{x}^{\dagger}$. 
\end{prop}

In the proof of Proposition \ref{prop: RoutingOptimalIncentive}, we first show that the externality function $\exterPop$ is monotonic in $\strategyPop$. Thus, the existence and uniqueness of fixed point toll price follows from Propositions \ref{prop: Alignment} and \ref{prop: UniquenessOfPInitFin}. Additionally, we show that the value of externality in equilibrium $\tilde{e}(\xEqPop{}{}(p))$ satisfies (C1) in Theorem \ref{thm: ConvergenceFin}. Therefore, we can conclude that the discrete-time learning dynamics converge to the socially optimal outcome. 

\section{Conclusion}\label{sec: Conclusion}
We propose a joint strategy and incentive update scheme for atomic and nonatomic games so that the emergent Nash equilibrium minimizes a social planner's cost (or equivalently maximizes social welfare). We assume that the planner, at each time-step, can modify the costs of players by setting a payment. There are three key features of the proposed scheme: first, the incentives are updated at a slower timescale as compared to the players' strategy update. Second, the incentive update is based on the externality caused by the players' strategy evaluated as the difference between players' marginal cost and the planner's marginal cost. Third, the incentive update is agnostic to the specific strategy update deployed by players, and relies on the current strategy profile. 

We show that the fixed point of the incentive and strategy update corresponds to an optimal incentive which induces a Nash equilibrium that is socially optimal. We provide sufficient conditions under which the proposed dynamic updates converge to its fixed points. We study the behavior of the proposed incentive updates to three games of practical significance: atomic networked quadratic aggregative games, atomic Cournot competition and nonatomic network routing games by verifying the proposed sufficient conditions of convergence.

\appendix 
\section{Proofs}

\subsection{Proof of Proposition \ref{prop: Alignment}}

{\begin{proof} We provide a detailed proof for the setting of atomic games as the proof for the non-atomic game follows similarly.

\noindent \textbf{Atomic game \(\finGame\):} We show that \(\Peq\) is non-empty. That is, there exists \(\pEqFin{}\) such that \(\exterFin(\xEqFin{}(\pEqFin{}))=\pEqFin{}\). Define a function \(\theta(p) = \exterFin(\xEqFin{}(p))\). Thus, the remaining proof is based on application of Brouwer's fixed point theorem to show existence of the fixed points of function \(\theta(\cdot)\).

We note that \(\theta(p)\) is a continuous function based on the setup presented in Sec \ref{sec: Model}.
 Furthermore, let's
define \(K\defas \{\theta(p): p\in \R^{|\playerSet|}\}\subset \R^{|\playerSet|}\). We claim that the set \(K\) is compact.
Indeed, this follows by two observations. First, the externality function \(\exterFin(\cdot)\) is continuous. Second,  the range of the function \(\xEqFin{}(\cdot)\) is \(\strategySetFin\) which is a compact space.  These two observations ensure that \(\theta(p)= \exterFin(\xEqFin{}(p))\) is a bounded function. 
Let \(\tilde{K}\defas \textsf{conv}(K)\) be the convex hull of \(K\), which in turn is also a compact set. Let's denote the restriction of function \(\theta\) on the set \(\tilde{K}\) as \(\theta_{|\tilde{K}}:\tilde{K}\ra \tilde{K}\) where \(\theta_{|\tilde{K}}(p)=\theta(p)\) for all \(p\in\tilde{K}\). We note that \(\theta_{|\tilde{K}}\) a is continuous function from a convex compact set to itself and therefore Brouwer's fixed point theorem ensures that there exits \(\pEqFin{}\in\tilde{K}\) such that \(\pEqFin{} =\theta_{|\tilde{K}}(\pEqFin{})= \theta(\pEqFin{})\). This concludes the proof about existence of \(\pEqFin{}\). 

Next, we show that incentive \(\pEqFin{}\) aligns Nash equilibrium with social optimality (i.e. for any \(\pEqFin{}\in\Peq\), \(\xEqFin{}(\pEqFin{})=\socOptFin\)).  
Fix \(\pEqFin{}\in\Peq\). For every \(i\in \playerSet\) we have \(\pEqFin{i}=\exterFin_i(\xEqFin{}(\pEqFin{}))\). This implies
\(
    \Der_{\strategyFin_i} \ell_i(\xEqFin{}(\pEqFin{}))+\pEqFin{i} = \Der_{\strategyFin_i} \socCostFin(\xEqFin{}(\pEqFin{}))\) for every \(i\in\playerSet\). This implies
    \begin{align}
       \label{eq: AlignmentSocNash} \JacobianIncentiveFin(\xEqFin{}(\pEqFin{}),\pEqFin{}) = \grad \socCostFin(\xEqFin{}(\pEqFin{})).
    \end{align}

Next, from \eqref{eq: FiniteGameVI} we know that \(\xEqFin{}(\pEqFin{})\) is a Nash equilibrium if and only if 
\begin{align}\label{eq: NASHVIProof}
\langle \JacobianIncentiveFin(\xEqFin{}(\pEqFin{}),\pEqFin{}), \strategyFin- \xEqFin{}(\pEqFin{})  \rangle \geq 0, \quad \forall \ \strategyFin \in \strategySetFin. 
\end{align}
Using \eqref{eq: AlignmentSocNash} and \eqref{eq: NASHVIProof}  the following holds:
\begin{align}\label{eq: SOCFINProof}
\langle \nabla \socCostFin(\xEqFin{}(\pEqFin{})), \strategyFin-\xEqFin{}(\pEqFin{})  \rangle \geq 0, \quad \forall \ \strategyFin \in \strategySetFin.
\end{align}
Comparing \eqref{eq: SOCFINProof} with \eqref{eq: SocCostFiniteVI} we note that \(\xEqFin{}(\pEqFin{})\) is the minimizer of social cost function \(\socCostFin\). This implies \(\xEqFin{}(\pEqFin{})=\socOptFin\) as \(\socOptFin\) is the unique minimizer of social cost function \(\socCostFin\).

\end{proof}

}

\subsection{Proof of Proposition \ref{prop: UniquenessOfPInitFin}}
{\begin{proof} The proof is based on a contradiction argument. 

\begin{itemize}
    \item[(i)] We make the following observation which are central to the proof:
    \begin{itemize}
    \item[\textbf{(O1)}] We note that if  \(\xEqFin{}(p)\in\textsf{int}(X)\) for every \(p\) then the variational inequality characterization \eqref{eq: FiniteGameVI} implies that \(\JacobianIncentiveFin(\xEqFin{}(p),p)=0\) for every \(p\). As a result the externality function \eqref{eq: ExterFin} becomes \(\exterFin(\xEqFin{}(p))=\grad \socCostFin(\xEqFin{}(p))+p\).  
    \item[\textbf{(O2)}] The strict convexity of the social cost function implies that 
    \begin{align*}
        \lara{ \grad \socCostFin(x)-\grad \socCostFin(y),x-y } > 0, \quad \forall \ x,y\in X \ \text{such that} \ x\neq y
    \end{align*}
    \end{itemize}
    
    \noindent{}Suppose there exists two distinct element \(p^\dagger,q^\dagger\in \Peq\). 
    We claim that \(\xEqFin{}(p^\dagger) \neq \xEqFin{}(q^\dagger)\). Indeed, using \textbf{(O1)} and \eqref{eq: GameJacobian} the following holds: 
    \begin{equation}
    \begin{aligned}\label{eq: DistinctJac}
        &p^\dagger_i = -\Der_{x_i}\lossFin(\xEqFin{}(p^\dagger)) ,\quad \forall \ i\in \playerSet\\
        &q^\dagger_i = -\Der_{x_i}\lossFin(\xEqFin{}(q^\dagger)), \quad \forall \ i\in\playerSet.
    \end{aligned}
    \end{equation}
     If \(\xEqFin{}(p^\dagger) = \xEqFin{}(q^\dagger)\) then \eqref{eq: DistinctJac}
    implies \(p^\dagger=q^\dagger\), but these are assumed to be distinct. Thus in the following proof we assume \(\xEqFin{}(p^\dagger)\neq \xEqFin{}(q^\dagger)\).
    
    We note from \textbf{(O1)} that
    \begin{align}\label{eq: DistinctPeq}
        0 = \grad \socCostFin(\xEqFin{}(p^\dagger)), \quad 0 = \grad \socCostFin(\xEqFin{}(q^\dagger)).
    \end{align}
    Substracting the two expressions in \eqref{eq: DistinctPeq} and taking inner product with \(\xEqFin{}(p^\dagger) - \xEqFin{}(q^\dagger)\) we see that 
    \begin{align}
      0 = \lara{\xEqFin{}(p^\dagger) - \xEqFin{}(q^\dagger), \grad \socCostFin(\xEqFin{}(p^\dagger))- \grad \socCostFin(\xEqFin{}(q^\dagger))} 
    \end{align}
    We arrive at a contradiction by noting that \(\xEqFin{}(p^\dagger)\neq \xEqFin{}(q^\dagger)\) and \textbf{(O2)} imply that RHS is strictly positive.
    \item[(ii)] To begin the proof we define \(\Der\ell(x) = (\Der_{x_i}\ell_i(x))_{i\in\playerSet}\). Under this notation, we have \(\exterFin(x) = \grad \socCostFin(x) - \Der \ell(x)\).
    The proof is based on the following observations:
    \begin{itemize}
        \item[\textbf{(O3)}] We claim that  \(\lara{\xEqFin{}(p_1)-\xEqFin{}(p_2),p_1-p_2}<0\) for any two distinct incentives \(p_1\neq p_2\). Indeed, from the variational inequality characterization of Nash equilibrium \eqref{eq: FiniteGameVI} we know that 
        \begin{align*}
    &\langle \Der \lossFin(\xEqFin{}(p_1))+p_1,x_1-\xEqPop{}{}(p_1) \rangle \geq 0, \quad \forall x_1\in\strategySetFin\\
    &\langle \Der \lossFin(\xEqFin{}(p_2))+p_2,x_2-\xEqPop{}{}(p_2) \rangle \geq 0, \quad \forall x_2\in\strategySetFin
\end{align*}
Picking \(x_1=\xEqPop{}{}(p_2)\) and \(x_2=\xEqPop{}{}(p_1)\), and adding the two inequalities in preceding equation we obtain 
\begin{align*}
    \langle\xEqFin{}(p_1)-\xEqFin{}(p_2),p_1-p_2 \rangle \leq -\langle \Der\lossFin(\xEqFin{}{}(p_1)) - \Der\lossFin(\xEqFin{}(p_2)),\xEqFin{}(p_1)-\xEqFin{}(p_2) \rangle \leq  0 
\end{align*}
where the last inequality follows due to the convexity of \(\ell\). 
    \end{itemize}
    
 We prove the uniqueness by contradiction. Suppose there exists two incentives \(p^\dagger,q^\dagger\in\Peq\) such that \(\exterPop(\xEqPop{}{}(p^\dagger))=p^\dagger\) and \(\exterPop(\xEqPop{}{}(q^\dagger))=q^\dagger\). Then we have   
\begin{align*}
   p^\dagger &= \grad \socCostFin(\xEqFin{}(p^\dagger))-\Der\lossFin(\xEqFin{}(p^\dagger)) \\ 
    q^\dagger &= \grad \socCostFin(\xEqFin{}(q^\dagger))-\Der\lossFin(\xEqFin{}(q^\dagger)). 
\end{align*}
Subtracting the two expressions and taking inner product with \(\xEqFin{}(p^\dagger)-\xEqFin{}(q^\dagger)\) we have
\[
\lara{\xEqFin{}(p^\dagger)-\xEqFin{}(q^\dagger),p^\dagger-q^\dagger} = \lara{\xEqFin{}(p^\dagger)-\xEqFin{}(q^\dagger), \exterFin(\xEqFin{}(p^\dagger)) - \exterFin(\xEqFin{}(q^\dagger)) } > 0.
\]
But from \textbf{(O3)} we see that we arrive at a contradiction as \(\lara{\xEqFin{}(p^\dagger)-\xEqFin{}(q^\dagger),p^\dagger-q^\dagger}\leq 0\). 
\end{itemize}

\end{proof}
}
\subsection{Proof of Theorem \ref{thm: ConvergenceFin}}
\begin{proof}
To ensure the convergence of \((x_k,p_k)\) to the fixed point \((x^\dagger,p^\dagger)\) of \eqref{eq: FinUpdateX}-\eqref{eq: FinUpdateP},  we exploit the timescale separation introduced due to Assumption \ref{assm: StepSizeAssumption}. The proof is based on two-timescale dynamical systems theory described in \cite{borkar1997stochastic}. Due to this timescale separation the strategy update evolves faster than the incentive update. This allows us to appropriately decouple the strategy and incentive update  and analyze them separately. 

Note that we can equivalently write \eqref{eq: FinUpdateX}-\eqref{eq: FinUpdateP} as follows:
\begin{equation}\label{eq: EqUpdate}
\begin{aligned}
    \strategyFin_{k+1} &= \strategyFin_k + \stepX_k\lr{ \stratUpdateFin(\strategyFin_k,\incentiveFin_k) - \strategyFin_k } \\
    \incentiveFin_{k+1} &= \incentiveFin_k + \stepX_k \lr{ \frac{\stepP_k}{\stepX_k} \lr{\exterFin(\strategyFin_k)-\incentiveFin_k} },
\end{aligned}
\end{equation}
where \(\lim_{k\ra\infty}\frac{\stepP_k}{\stepX_k}=0\) and \(\lim_{k\ra\infty} \stepX_k = 0\). From two timescale dynamical systems theory  we know that under Assumption \ref{assm: ConvergenceStrategy}  the tuple \((x_k,p_k)\) converges to the set \(\{(\xEqFin{}(\incentiveFin),\incentiveFin):\incentiveFin\in\R^{|\playerSet|}\}\). 
Thus for sufficiently large values of \(k\), the update \(\strategyFin_k\) closely tracks \( \xEqFin{}(\incentiveFin_{k})\). Therefore, we consider the following update to analyze the convergence of the slow incentive update \eqref{eq: FinUpdateP}:
\begin{align}\label{eq: SlowDisFin}
    \incentiveFinSep_{k+1} &= (1-\stepp{k})\incentiveFinSep_{k}+\stepp{k} \mdFin(\xEqFin{}(\incentiveFinSep_k)).
\end{align}
Since the step sizes \(\{\stepP_k\}\) are asymptotically going to zero and is non-summable (Assumption \ref{assm: StepSizeAssumption}-(i)) we can approximate the updates in \eqref{eq: SlowDisFin} by the following continuous-time dynamical system :
\begin{align}
    \dot{\incentiveFinSep}(t) &=  \exterFin(\xEqFin{}(\incentiveFinSep(t))) - \incentiveFinSep(t), \label{subeq:p_fin1}
\end{align}
Convergence of discrete-time updates \eqref{eq: FinUpdateX}-\eqref{eq: FinUpdateP} then hold if the flow of \eqref{subeq:p_fin1} globally converges to \(\Peq\). 

Requirements (C1) in Theorem \ref{thm: ConvergenceFin} is a sufficient condition for convergence of the trajectories of \eqref{subeq:p_fin1} to the set \(\Peq\). This condition is based on cooperative dynamical systems theory \cite{hirsch1985systems}.  
 On the other hand requirement (C2) in Theorem \ref{thm: ConvergenceFin} ensures convergence the trajectories of \eqref{subeq:p_fin1} to the set \(\Peq\) by demanding existence of a Lyapunov function \cite{sastry2013nonlinear} that is strictly positive everywhere except at \(\Peq\) and decreases along the flow of {\eqref{subeq:p_fin1}}.


\end{proof}
\subsection{Proof of Proposition \ref{prop: QuadIncentive}}

\begin{proof}
We first show that the set \(\Peq\) is singleton for the setup in Sec \ref{ssec: Quad}\footnote{Note that we cannot directly use Proposition \ref{prop: Alignment} as that required compactness of strategy space.}.  Then we show that the dynamic update \((x_k,p_k)\) corresponding to \eqref{eq: strategyQuad}-\eqref{eq: PriceQuad} converges to the social optimality \((x^\dagger,p^\dagger)\). 

Note that any element \(p^\dagger\in \Peq\) should satisfy \(p^\dagger = \exterFin(\xEqFin{}(p^\dagger)) = \xi+Aw\xEqFin{}(p^\dagger)\). Moreover, from \eqref{eq: NashQuad} we know that \(\xEqFin{}(\pEqFin{}) = Aw\xEqFin{}(p^\dagger)-p^\dagger\). 
Succintly writing the preceding two relations in matrix form gives us:
\begin{align*}
    \underbrace{\begin{bmatrix}
    I & -\Kmat\Zmat \\ I & I -\Kmat\Zmat
    \end{bmatrix}}_{\Gamma}\begin{bmatrix}
    \pEqFin{}\\ \xEqFin{}(\pEqFin{})
    \end{bmatrix} = \begin{bmatrix}
    \Bmat \\ 0
    \end{bmatrix}.
\end{align*}
We claim that \(\Gamma\) is an invertible matrix\footnote{Invertibility of \(I-Aw\) is a necessary condition for invertiblility of \(\Gamma\).} with the inverse as follows:
\begin{align*}
    \Gamma^{-1} = \begin{bmatrix}
    I-Aw & Aw \\ -I & I
    \end{bmatrix}
\end{align*}
Thus \((\xEqFin{}(\pEqFin{}),\pEqFin{})\) exists and is unique. Moreover \(\pEqFin{} = (I-\Kmat\Zmat)\Bmat\) and \(\xEqFin{}(\pEqFin{}) = -\Bmat= \socOptFin\).  

Next, to ensure that the dynamic update \((x_k,p_k)\) corresponding to \eqref{eq: NashQuad}-\eqref{eq: PriceQuad} converges to the fixed point \((\socOptFin,\pEqFin{})\) we use Theorem \ref{thm: ConvergenceFin}. It is sufficient to show that Assumption \ref{assm: ConvergenceStrategy} and condition \(\textbf{(C2)}\) hold in order to use the results from Theorem \ref{thm: ConvergenceFin} directly.

First, we show that Assumption \ref{assm: ConvergenceStrategy} hold. That is, for any fixed incentive update \((\incentiveFin_k)\equiv p\) the strategy update satisfies \(\lim_{k\ra\infty}x_k = \xEqFin{}(p)\). Indeed, due to Assumption \ref{assm: StepSizeAssumption}, the convergence properties of discrete time updates can be obtained by analysing the corresponding continuous time dynamical system.
That is, we consider the following continuous time dynamical system corresponding to the strategy update:
\begin{align}\label{eq: ContSysQuadStrat}
    \dot{\strategyFinSep}(t) = -(I-\Kmat\Zmat) \strategyFinSep(t)-p.
\end{align}
We note that the trajectories of \eqref{eq: ContSysQuadStrat} satisfy \(\lim_{t\ra\infty} \strategyFinSep(t) = \xEqFin{}(p)\). This is due to the assumption that \(-(I-\Kmat\Zmat)\) is Hurwitz\footnote{A matrix \(A\) is called Hurwitz if \(\spec(A)\subset \C_{-}^{\circ}\).} \cite{callier2012linear}. Thus Assumption \ref{assm: ConvergenceStrategy} holds.

Next, we show that condition \textbf{(C2)} is satisfied which then fulfils all the requirement of Theorem \ref{thm: ConvergenceFin}. We claim that the function \(V(\incentiveFin) = (\incentiveFin-\pEqFin{})^\top L(\incentiveFin-\pEqFin{})\) satisfies \(\textbf{(C2)}\) where \(L\)\footnote{Note that the existence of such matrix \(L\) is guaranteed from Lyapunov theorem \cite{callier2012linear} as \(\sigma(I-\Kmat\Zmat) \subset \C_+^{\circ}\).} is a symmetric positive definite matrix that satisfies the following condition: 
\begin{align}\label{eq: LyapEq}
    (I-\Kmat\Zmat)^{-\top}M + M (I-\Kmat\Zmat)^{-1} = I.
\end{align}
Indeed, \(V(\pEqFin{}) = 0\) and since \(L\) is a positive definite matrix, this means \(V(p)>0\) for all \(p\neq \pEqFin{}\). Furthermore, we compute 
\begin{align*}
    &\grad V(p)^\top (\exterFin(\xEqFin{}(p))-p) = 2(p-\pEqFin{})^\top L (\exterFin(\xEqFin{}(p))-p), \\
    &\underset{(a)}{=} 2(p-\pEqFin{})^\top L \lr{ \Bmat+\Kmat\Zmat \xEqFin{}(p)-p }, \\
    &= 2(p-\pEqFin{})^\top L \lr{ -\xEqFin{}(\pEqFin{})+\Kmat\Zmat \xEqFin{}(p)-p }, \\ 
    &\underset{(b)}{=} 2(p-\pEqFin{})^\top L \lr{ -\xEqFin{}(\pEqFin{})+\xEqFin{}(p) }, \\
    &\underset{(c)}{=} -2(p-\pEqFin{})^\top L(I-\Kmat\Zmat)^{-1}(p-\pEqFin{}), \\ 
    &\underset{(d)}{=} -{(p-\pEqFin{})^\top \lr{L(I-\Kmat\Zmat)^{-1} + (I-\Kmat\Zmat)^{-\top}L}(p-\pEqFin{})  }, \\ 
    &= -(p-\pEqFin{})^\top (p-\pEqFin{}) < 0,
\end{align*}
where \((a)\) is by the definition of externality function \eqref{eq: ExterFin}, \((b),(c)\) is by the Nash equilibrium \eqref{eq: NashQuad} and \((d)\) is by \eqref{eq: LyapEq}.

    This completes the proof.
\end{proof}
\subsection{Proof of Proposition \ref{prop: CournotConvergence}}
Before stating the proof of Proposition \ref{prop: CournotConvergence} we present the following two results which are crucial in the proof of Proposition \ref{prop: CournotConvergence}. First, we prove the Nash equilibrium takes the form stated in \eqref{eq: NashCournot}. Next, we present a technical lemma.  

Below, we state the Nash equilibrium in Cournot competition in terms of incentives. 
\begin{lemma}[Nash equilibrium]\label{lem: NashCournot}
For any given incentive \(p\), the Nash equilibrium is given by 
\begin{align*}
    \xEqFin{}(p) = \frac{\theta-\nu}{\delta(|\playerSet|+1)} \mathbbm{1} - \frac{1}{\delta}p+\frac{1}{\delta(|\playerSet|+1)}\mathbbm{1}\mathbbm{1}^\top p 
\end{align*}
\end{lemma}
\begin{proof}
In the setup of Sec \ref{ssec: Cournot} the variational inequality characterization of Nash equilibrium \eqref{eq: FiniteGameVI} implies that for any given \(p\), \(\xEqFin{}(p)\) is a Nash equilibrium if and only if $J(x^*(p), p ) = 0$. Consequently, \(\xEqFin{}(p)\) satisfies the following
\begin{align*}
    2x_i^*(p) + \sum_{j \neq i} x_j^*(p) = \frac{\theta - \nu - p_i}{\delta}.
\end{align*}
Recasting this in the matrix form gives the following:
\begin{align}\label{eq: MatNash}
    \underbrace{\begin{bmatrix}
    2 & 1 & \dots & 1 \\ 1 & 2 & \dots & 1 \\ \dots & \dots & \dots & \dots \\1  & 1 & \dots & 2
    \end{bmatrix}}_{A} \begin{bmatrix}
    x_1^*(p) 
    \\ x_2^*(p) \\ \dots \\ x_n^*(p)
    \end{bmatrix} = \begin{bmatrix} \frac{\theta - \nu - p_1}{\delta}  \\ \frac{\theta - \nu - p_2}{\delta}  \\ \dots \\ \frac{\theta - \nu - p_n}{\delta} \end{bmatrix}.
\end{align}
Note that $A = I + \mathbbm{1} \mathbbm{1}^{\top}$. Furthermore, by the
Sherman-Morrison formula:
\begin{align}\label{eq: ShermanBlock}
    A^{-1}  
    & = \frac{1}{|\mathcal{I}| +1} \begin{bmatrix} |\mathcal{I}|  
     & -1 & \dots & -1 \\ -1  
     & |\mathcal{I}| & \dots & -1 \\ \dots
     & \dots & \dots & \dots \\ -1  
     & -1 & \dots & |\mathcal{I}|
    \end{bmatrix} = I - \frac{1}{|\playerSet|+1}\mathbbm{1}\mathbbm{1}^\top.
\end{align}
Therefore, by \eqref{eq: MatNash} we have :
\begin{align*}
    \begin{bmatrix}
    x_1^*(p) 
    \\ x_2^*(p) \\ \dots \\ x_n^*(p)
    \end{bmatrix} = \frac{1}{\delta (|\mathcal{I}|+1)} \begin{bmatrix}
   \theta-\nu - |\mathcal{I}| p_1  +\sum_{j\neq 1} p_j \\ \theta-\nu - |\mathcal{I}| p_2 +\sum_{j\neq 2} p_j \\ \dots \\ \theta-\nu - |\mathcal{I}| p_n + \sum_{j\neq n} p_j
    \end{bmatrix}.
\end{align*}
This completes the proof.
\end{proof}

Next, we present a technical lemma that is crucial in the proof of Proposition \ref{prop: CournotConvergence}. 
\begin{lemma}\label{lem: GammaOmega}
Let $\Gamma = (2\lambda - \delta) I +\delta \mathbbm{1} \mathbbm{1}^{\top} $ and $\Omega = - \frac{1}{\delta}I + \frac{1}{\delta (|\mathcal{I}| + 1)}  \mathbbm{1}\mathbbm{1}^{\top} $. If \(\lambda > \delta\) then \(\spec(\Gamma\Omega)\subset \C_{-}^\circ\).
\end{lemma}
\begin{proof}
Note that 
\begin{align*}
    -\Gamma \Omega &= \frac{1}{\delta (|\playerSet|+1)} ((2\lambda -\delta)I + \delta \mathbbm{1} \mathbbm{1}^{\top})((|\playerSet|+1)I-\mathbbm{1} \mathbbm{1}^{\top}) \\ 
    &= \frac{1}{\delta (|\playerSet|+1)}\lr{(|\playerSet|+1)(2\lambda - \delta)I - (2\lambda -\delta)\mathbbm{1} \mathbbm{1}^{\top} + (|\playerSet|+1)\delta \mathbbm{1} \mathbbm{1}^{\top} - |\playerSet|\delta \mathbbm{1} \mathbbm{1}^{\top} }\\ 
    &=\frac{1}{\delta (|\playerSet|+1)} \lr{ (|\playerSet|+1)(2\lambda - \delta)I - (2\lambda -2\delta)\mathbbm{1} \mathbbm{1}^{\top}}
\end{align*}
From Gershgorin's circle theorem\footnote{Gershgorin's circle theorem \cite{varga2010gervsgorin}, which says that for a square matrix $A \in \mathbb{R}^{n \times n}$, each eigenvalue of $A$ is contained in at least one of the disks:
\begin{align}
    D_i = \{ z\in \mathbb{C} : |z-A_{ii}| \leq \sum_{j \neq i} |A_{ij}| \}
\end{align}
where $A_{ii}$ are the diagonal entries of $A$, and $A_{ij}$ are the off-diagonal entries. In our case, to ensure that $\Gamma \Omega$ has eigenvalues on the open right half plane, we need to ensure that $|A_{ii}|- \sum_{j \neq i} |A_{ij}| > 0$} we know that for the result to hold it is sufficient to ensure 
\begin{align}\label{eq: GershCournot}
    |(|\playerSet|+1)(2 \lambda - \delta) - (2 \lambda - 2 \delta)| > (|\playerSet|-1) |2 \lambda - 2 \delta| 
\end{align}

In fact if \(\lambda>\delta\) then \eqref{eq: GershCournot} holds. 

\end{proof}


Finally, we present the proof of Proposition \ref{prop: CournotConvergence} below:
\begin{proof}[Proof of Proposition \ref{prop: CournotConvergence}]
We show that the set \(\Peq\) is singleton for the setup in Sec \ref{ssec: Quad}\footnote{Note that we cannot directly use Proposition \ref{prop: Alignment} as that required compactness of strategy space.}.  Then we show that the dynamic update \((x_k,p_k)\) corresponding to \eqref{eq: strategyQuad}-\eqref{eq: PriceQuad} converges to social optimality \((x^\dagger,p^\dagger)\). 

Note that any element \(p^\dagger\in \Peq\) should satisfy \(p^\dagger = \exterFin(\xEqFin{}(p^\dagger))=\lr{(2\lambda - \delta)I+\delta \mathbbm{1}\mathbbm{1}^\top}\xEqFin{}(p^\dagger)\). Moreover from Lemma \ref{lem: NashCournot} we know that \(\xEqFin{}(p^\dagger) = \frac{\theta-\nu}{\delta(|\playerSet|+1)} \mathbbm{1} - \frac{1}{\delta}\lr{I-\frac{\mathbbm{1}\mathbbm{1}^\top}{|\playerSet|+1}}p^\dagger\). Succintly writing these two requirements in matrix form gives us:
\begin{align*}
    \underbrace{\begin{bmatrix}
    (2\lambda-\delta) I + \delta \mathbbm{1}\mathbbm{1}^\top & -I \\ 
    \delta I  & I-\frac{1}{|\playerSet|+1}\mathbbm{1}\mathbbm{1}^\top
 \end{bmatrix}}_{B}\begin{bmatrix}
 \xEqFin{}(p^\dagger) \\ p^\dagger
 \end{bmatrix} = \begin{bmatrix}
 0 \\ \frac{\theta-\nu}{|\playerSet|+1}\mathbbm{1}
 \end{bmatrix}
\end{align*}
We claim that \(B\) is an invertible matrix.  Indeed, lower diagonal is an invertible block by \eqref{eq: ShermanBlock} and the Schur complement of B with respect to that block is \(2\lambda I + 2\delta \mathbbm{1}\mathbbm{1}^\top\) which is also invertible. Thus \((\xEqFin{}(p^\dagger),p^\dagger)\) exits and is unique. 



Next,  we use Theorem \ref{thm: ConvergenceFin} to ensure that the dynamic update \((x_k,p_k)\) corresponding to \eqref{eq: CournotStrat}-\eqref{eq: CournotIncentive} converges to the fixed point \((\xEqFin{}(p^\dagger),p^\dagger)\). It is sufficient to show that Assumption \ref{assm: ConvergenceStrategy} and condition \textbf{(C2)} hold. Before checking these conditions we define $\Gamma = (2\lambda - \delta) I +\delta \mathbbm{1} \mathbbm{1}^{\top} $ and $\Omega = - \frac{1}{\delta}I + \frac{1}{\delta (|\mathcal{I}| + 1)}  \mathbbm{1}\mathbbm{1}^{\top} $.

First, we show that Assumption \ref{assm: ConvergenceStrategy} holds. That is, for any fixed incentive $(p_k) = p$, the strategy update satisfies $\lim_{ k \rightarrow \infty} x_k = x^*(p)$. Indeed, due to Assumption \ref{assm: StepSizeAssumption}, the convergence properties of discrete time updates can be obtained by analysing the corresponding continuous time dynamical system stated below:
\begin{align}\label{eq: CournotBestResponse}
    \dot{\strategyFinSep}_i(t) = - \strategyFinSep_i(t) + \left( \frac{\theta - \delta \sum_{j \neq i} \strategyFinSep_{j}(t) -\nu - p_i}{2 \delta} \right)  
\end{align}
We claim that the trajectories of \eqref{eq: CournotBestResponse} satisfy $\lim_{t \rightarrow 0} \strategyFinSep(t) = x^*(p)$. Indeed, the atomic Cournot competition is a potential game with the following potential function for any \(p\):
\begin{align}
    T(\strategyFinSep, p) = - \int_{0}^{\sum_{i=1}^{|\mathcal{I}|} \strategyFinSep_i } (\theta - \delta z) dz + \sum_{i=1}^{|\mathcal{I}|} (\nu + p_i) x_i,
\end{align}
and \eqref{eq: CournotBestResponse} is the corresponding continuous time best response dynamics. Thus \cite[Theorem 2]{swenson2018best} ensures $\lim_{t \rightarrow \infty} \strategyFinSep(t) = x^*(p)$. 

Next, we show that condition $\textbf{(C2)}$ is satisfied, which fulfills the requirements of Theorem \eqref{thm: ConvergenceFin}.
We claim that the function $V(p)= (p-p^{\dagger}) L (p-p^{\dagger})$ satisfies $\textbf{(C2)}$, where $L$ is a symmetric positive definite matrix that satisfies:
\begin{align}\label{eq: LyapCournot}
    (\Gamma \Omega)^{\top} L + L^{\top} (\Gamma \Omega) = -I
\end{align}
Note that the existence of \(L\) follows from the Lyapunov theorem \cite{callier2012linear} as from Lemma \ref{lem: GammaOmega} we know that \(-\Gamma\Omega\) is a Hurwitz matrix. 

Indeed, $V(p^{\dagger}) = 0$ and since $L$ is positive definite, this means $V(p) > 0$ for all $p \neq p^{\dagger}$. Furthermore, we compute:
\begin{align*}
    \nabla V(p)^{\top} (e(x^*(p)) - p) &= 2 (p-p^{\dagger})^{\top} L (e(x^*(p)) - p) \\ & \underset{(a)}{=} 2 (p-p^{\dagger})^{\top} L \lr{\lr{(2 \lambda - \delta) I + \delta \mathbbm{1} \mathbbm{1}^{\top}} x^*(p) - p} \\ 
    & \underset{(b)}{=} 2 (p-p^{\dagger})^{\top} L \lr{\lr{(2 \lambda - \delta)I + \delta \mathbbm{1} \mathbbm{1}^{\top}}(x^*(p) -x^*(p^{\dagger})) + p^{\dagger} - p} \\ 
    & =  2 (p-p^{\dagger})^\top L \Gamma (x^*(p)- x^*(p^{\dagger})) -2(p-p^{\dagger})^\top L (p-p^{\dagger})\\ & \underset{(c)}{=}  2 (p-p^{\dagger})^\top L \Gamma \Omega (p-p^{\dagger}) -2(p-p^{\dagger})^\top L (p-p^{\dagger}) \\
    & \underset{(d)}{=}  (p-p^{\dagger})^\top (L (\Gamma \Omega) + (\Gamma \Omega)^{\top} L) (p-p^{\dagger})-2(p-p^{\dagger})^\top L (p-p^{\dagger})  \\ & =- (p-p^{\dagger})(2L+I) (p-p^{\dagger})
 \end{align*}
 where $(a)$ is by the definition of the externality function $e(x^*(p))$, $(b)$ is by adding and subtracting $p^{\dagger}$, $(c)$ is by the definition of the Nash equilibrium $x^*(p)$, and $(d)$ is by \eqref{eq: LyapCournot}. 
 This completes the proof. 
 \end{proof}
\bibliography{refs}

\begin{thebibliography}{10}

\bibitem{barrera2014dynamic}
Jorge Barrera and Alfredo Garcia.
\newblock Dynamic incentives for congestion control.
\newblock {\em IEEE Transactions on Automatic Control}, 60(2):299--310, 2014.

\bibitem{bacsar1984affine}
Tamer Ba{\c{s}}ar.
\newblock Affine incentive schemes for stochastic systems with dynamic
  information.
\newblock {\em SIAM Journal on Control and Optimization}, 22(2):199--210, 1984.

\bibitem{borkar1997stochastic}
Vivek~S Borkar.
\newblock Stochastic approximation with two time scales.
\newblock {\em Systems \& Control Letters}, 29(5):291--294, 1997.

\bibitem{cai2019role}
Desmond Cai, Subhonmesh Bose, and Adam Wierman.
\newblock On the role of a market maker in networked cournot competition.
\newblock {\em Mathematics of Operations Research}, 44(3):1122--1144, 2019.

\bibitem{callier2012linear}
Frank~M Callier and Charles~A Desoer.
\newblock {\em Linear system theory}.
\newblock Springer Science \& Business Media, 2012.

\bibitem{como2021distributed}
Giacomo Como and Rosario Maggistro.
\newblock Distributed dynamic pricing of multiscale transportation networks.
\newblock {\em IEEE Transactions on Automatic Control}, 2021.

\bibitem{davis1962externalities}
Otto~A Davis and Andrew Whinston.
\newblock Externalities, welfare, and the theory of games.
\newblock {\em Journal of Political Economy}, 70(3):241--262, 1962.

\bibitem{facchinei2007finite}
Francisco Facchinei and Jong-Shi Pang.
\newblock {\em Finite-dimensional variational inequalities and complementarity
  problems}.
\newblock Springer Science \& Business Media, 2007.

\bibitem{fudenberg1998theory}
Drew Fudenberg and David Levine.
\newblock {\em The theory of learning in games}, volume~2.
\newblock MIT press, 1998.

\bibitem{harris1998rate}
Christopher Harris.
\newblock On the rate of convergence of continuous-time fictitious play.
\newblock {\em Games and Economic Behavior}, 22(2):238--259, 1998.

\bibitem{hirsch1985systems}
Morris~W Hirsch.
\newblock Systems of differential equations that are competitive or cooperative
  ii: Convergence almost everywhere.
\newblock {\em SIAM Journal on Mathematical Analysis}, 16(3):423--439, 1985.

\bibitem{ho1982control}
Yu-Chi Ho, Peter~B Luh, and Geert~Jan Olsder.
\newblock A control-theoretic view on incentives.
\newblock {\em Automatica}, 18(2):167--179, 1982.

\bibitem{jackson2015games}
Matthew~O Jackson and Yves Zenou.
\newblock Games on networks.
\newblock In {\em Handbook of game theory with economic applications},
  volume~4, pages 95--163. Elsevier, 2015.

\bibitem{levin1985taxation}
Dan Levin.
\newblock Taxation within cournot oligopoly.
\newblock {\em Journal of Public Economics}, 27(3):281--290, 1985.

\bibitem{liu2021inducing}
Boyi Liu, Jiayang Li, Zhuoran Yang, Hoi-To Wai, Mingyi Hong, Yu~Marco Nie, and
  Zhaoran Wang.
\newblock Inducing equilibria via incentives: Simultaneous design-and-play
  finds global optima.
\newblock {\em arXiv preprint arXiv:2110.01212}, 2021.

\bibitem{maheshwari2021dynamic}
Chinmay Maheshwari, Kshitij Kulkarni, Manxi Wu, and Shankar Sastry.
\newblock Dynamic tolling for inducing socially optimal traffic loads.
\newblock {\em arXiv preprint arXiv:2110.08879}, 2021.

\bibitem{mazumdar2020gradient}
Eric Mazumdar, Lillian~J Ratliff, and S~Shankar Sastry.
\newblock On gradient-based learning in continuous games.
\newblock {\em SIAM Journal on Mathematics of Data Science}, 2(1):103--131,
  2020.

\bibitem{monderer1996fictitious}
Dov Monderer and Lloyd~S Shapley.
\newblock Fictitious play property for games with identical interests.
\newblock {\em Journal of economic theory}, 68(1):258--265, 1996.

\bibitem{nachbar1990evolutionary}
John~H Nachbar.
\newblock “evolutionary” selection dynamics in games: Convergence and limit
  properties.
\newblock {\em International journal of game theory}, 19(1):59--89, 1990.

\bibitem{paccagnan2019incentivizing}
Dario Paccagnan, Rahul Chandan, Bryce~L Ferguson, and Jason~R Marden.
\newblock Incentivizing efficient use of shared infrastructure: Optimal tolls
  in congestion games.
\newblock {\em arXiv preprint arXiv:1911.09806}, 2019.

\bibitem{parise2017sensitivity}
Francesca Parise and Asuman Ozdaglar.
\newblock Sensitivity analysis for network aggregative games.
\newblock In {\em 2017 IEEE 56th Annual Conference on Decision and Control
  (CDC)}, pages 3200--3205. IEEE, 2017.

\bibitem{pigou2017economics}
Arthur~Cecil Pigou and Nahid Aslanbeigui.
\newblock {\em The economics of welfare}.
\newblock Routledge, 2017.

\bibitem{ratliff2020adaptive}
Lillian~J Ratliff and Tanner Fiez.
\newblock Adaptive incentive design.
\newblock {\em IEEE Transactions on Automatic Control}, 66(8):3871--3878, 2020.

\bibitem{sandholm2010population}
William~H Sandholm.
\newblock {\em Population games and evolutionary dynamics}.
\newblock MIT press, 2010.

\bibitem{sastry2013nonlinear}
Shankar Sastry.
\newblock {\em Nonlinear systems: analysis, stability, and control}, volume~10.
\newblock Springer Science \& Business Media, 2013.

\bibitem{swenson2018best}
Brian Swenson, Ryan Murray, and Soummya Kar.
\newblock On best-response dynamics in potential games.
\newblock {\em SIAM Journal on Control and Optimization}, 56(4):2734--2767,
  2018.

\bibitem{varga2010gervsgorin}
Richard~S Varga.
\newblock {\em Ger{\v{s}}gorin and his circles}, volume~36.
\newblock Springer Science \& Business Media, 2010.

\bibitem{varian1994solution}
Hal~R Varian.
\newblock A solution to the problem of externalities when agents are
  well-informed.
\newblock {\em The American Economic Review}, pages 1278--1293, 1994.

\end{thebibliography}
\bibliographystyle{plain}
\end{document}